\documentclass[%
 reprint,
 amsmath,amssymb,
 aps,
]{revtex4-2}

\usepackage{graphicx}
\usepackage{dcolumn}
\usepackage{bm}

\usepackage{booktabs,multirow,tabularx,array} 
\usepackage{acronym}
\usepackage{comment} 
\usepackage{physics}
\usepackage{enumitem}
\usepackage{xcolor}

\usepackage{relsize}

\usepackage{subcaption}

\usepackage{algorithm}
\usepackage[]{algpseudocode}

\newacro{SD}[SD]{strictly dissipative}
\newacro{ASD}[ASD]{almost strictly dissipative}

\newacro{SLO}[SL-operator]{Sturm-Liouville operator}

\newacro{DPS}[DPS]{distributed parameter systems}

\newacro{BCS}[BCS]{boundary control system}
\newacro{DBCS}[DBCS]{disturbed boundary control system}

\newacro{LTI}[LTI]{linear time-invariant}

\newacro{IC-POVM}[IC-POVM]{informationally-complete \ac{POVM}}
\newacro{SIC-POVM}[SIC-POVM]{symmetric informationally-complete \ac{POVM}}
\newacro{POVM}[POVM]{positive operator-valued measure}
\newacro{QIS}[QIS]{quantum information science}
\newacro{QSE}[QSE]{quantum state estimation}
\newacro{QST}[QST]{quantum state tomography}

\newacro{QSO}[QSO]{quantum state observer}

\newacro{QP}[QP]{quadratic program}
\newacro{KKT}[KKT]{Karush-Kuhn-Tucker}
\newacro{LASSO}[LASSO]{least absolute shrinkage and selection operator}
\usepackage{StyleFiles/mathnotation}
\usepackage{StyleFiles/notation}
\usepackage{StyleFiles/localnotation}

\begin{document}

\preprint{}

\title{A Note on the Estimation of Von Neumann and Relative Entropy\\
via Quantum State Observers}
\thanks{Corresponding author: Mark Balas, mbalas@tamu.edu}%

\author{Mark Balas$^1$, Vinod P. Gehlot$^2$, Tristan D. Griffith$^1$}%
\affiliation{%
 $^1$Department~of~Mechanical~Engineering,~Texas A\&M University,~College~Station,~TX,~USA\\
 $^2$Jet~Propulsion~Laboratory,~California~Institute~of~Technology,~Pasadena,~CA,~USA
}%

\date{January 5, 2023}

\begin{abstract}
An essential quantity in quantum information theory is the von Neumann entropy which depends entirely on the quantum density operator. Once known, the density operator reveals the statistics of observables in a quantum process, and the corresponding von Neumann Entropy yields the full information content. However, the state, or density operator, of a given system may be unknown. Quantum state observers have been proposed to infer the unknown state of a quantum system. In this note, we show (i) that the von Neumann entropy of the state estimate produced by our quantum state observer is exponentially convergent to that of the system's true state, and (ii) the relative entropy between the system and observer's state converges exponentially to zero as long as the system starts in a full-rank state. 
\end{abstract}

\maketitle


\section{\label{sec:Intro}Introduction:\protect\\}
Quantum probability theory~\cite{paul_dirac, von_neumannn, gudder} predicts the behavior of many physical systems with unprecedented accuracy.
%
When modelling such systems, it is important to be able to characterize their statistical properties.
The fundamental element of quantum statistical mechanics is the density operator. Once determined, this operator reveals the statistics of quantum observables in a dynamical quantum system. This note concentrates on the inference of the density operator for a generic closed quantum system described by a Liouville--von Neumann master equation. Prior work on linear observers for quantum systems was done in~\cite{BalGeh:22}. Developments in quantum measurement and the quantum Kalman filter appear in \cite{jacobs,bouten} with theoretical issues and foundations in \cite{accardi}. Linear observability of quantum systems was studied in~\cite{GriBal:22}. Clou\^{a}tr\'{e} and Balas et al. give a full characterization of the linear observability of closed quantum systems with \ac{POVM} output in~\cite{clouatre2022linear} and provide a canonical quantum state observer. Since the original work of~\cite{BalGeh:22}, the Hilbert metric projection of a square matrix onto the set of quantum density operators has been studied. In particular, it can be paired with a linear quantum state observer to produce an exponentially convergent nonlinear observer which ensures the observer's state is a valid quantum density operator. Recently~\cite{clouatre2023closedform} derived a closed-form solution of the metric projection making said nonlinear observers realizable. Now that these types of observers have been unlocked, questions regarding quantum information theoretic properties of the quantum state observer should be studied. Specifically, what quantities of information can be efficiently estimated for a quantum system in an unknown state? Answering this question could lead to additional practical uses of quantum state observers in quantum networks and quantum metrology which rely heavily on information-theoretic ideas~\cite{BelHir:13, TotGez:14}. There are numerous quantum information-theoretic quantities, \eg~the von Neumann entropy, quantum relative entropy, and quantum Fisher information.

In this note, we study the von Neumann entropy in the context of quantum state observers. In particular, we show that the von Neumann entropy of the observer's state always converges exponentially to the entropy of the system's state. We also study the quantum relative entropy. While there are cases where one may not expect the quantum relative entropy to converge, we show under reasonable assumptions (namely that the system starts in a full-rank state) that the quantum relative entropy between the observer's state and the system's state converges exponentially to zero; this means that the two states become indistinguishable. Proof of these facts rely on logarithmic inequalities like those developed by Fannes~\cite{Fannes1973}.

\section{Quantum state observers}\label{sec:LQSO}
Linear quantum state observers are used to estimate the state of a quantum dynamical system of the form
\begin{equation}\label{lindblad}
    \dot{\dens} = -\i\, [\Ham,\, \dens]
\end{equation}
where $\dens\in\H$ is the state of the system, $d$ is the dimension of the system, $\Ham\in\H$ is the system Hamiltonian, $[\M{A},\, \M{B}]$ is the commutator of $\M{A}$ and $\M{B}$, and  $\i\triangleq\sqrt{-1}$. The state $\dens$ belongs to the set $\S = \{\altdens\in\H\, :\, \trace{\altdens}=1,\, \altdens=\altdens^\dagger,\, \altdens\succeq\M{0}\}$ of valid density operators. If the initial (unknown) state of the system is $\dens_0\triangleq\dens(0)$, the goal of a quantum state observer is to produce a estimate $\hdens(t)$ which converges to the true state over time:
\begin{equation}\label{convergentError}
    \lim_{t\rightarrow\infty} \|\hdens(t) - \dens(t)\| = 0
\end{equation}
where $\|\cdot\|$ is the Hilbert-Schmidt norm on $\H$. This estimate is generated using a set of measurement statistics that are obtained over numerous experiments. The measurement statistics are summarized via a vector $\V{y}\in\R^{\numMeas}$ given by
\begin{equation}\label{out}
    \V{y}(t) = 
    \begin{bmatrix}
        \trace{\Meas_1\, \dens(t)}\\
        \trace{\Meas_2\, \dens(t)}\\
        \vdots\\
        \trace{\Meas_{\numMeas}\, \dens(t)}
    \end{bmatrix}
\end{equation}
where $\numMeas$ is the number of measurement statistics and $\{\Meas_{\measInd}\}_{\measInd=1}^{\numMeas} \subset \H$ are observables.

The first question to ask is whether the observer's task is feasible. This is certainly the case if the system is \textit{linearly observable}~\cite[sec. 14]{brockett}. If the observables $\{\Meas_{\measInd}\}_{\measInd=1}^{\numMeas}$ form a \ac{POVM}, then the linear observability problem reduces to that of Clou\^{a}tr\'{e} et al.~\cite{clouatre2022linear} and the observability can be easily tested using the dynamics in the form of~\eqref{lindblad}.

\begin{rem}
    As remarked in~\cite{clouatre2022linear}, without additional considerations for measurement back-action, the output $\V{y}(t)$ may not be obtained over any continuous interval due to the nature of obtaining measurement statistics from a quantum system. However, the observer may be discretized (sampled~\cite[pg. 116]{antsaklis}) to obtain an exponentially convergent discrete time observer. Developing theory in continuous time allows one to determine the limits of the theory. 
\end{rem}

Denote the vectorization of the density operator as $\V{x}\triangleq\Vec{\dens}\in\C^{d^2}$, which is a natural isometry between the Hilbert spaces $\C^{d^2}$ and $\H$ equipped with their canonical inner products. The dynamics of the vectorized system can now be efficiently written as
\begin{equation}\label{linearSys}
    \begin{cases}
      \dot{\V{x}}(t) = \M{A}\, \V{x}(t)\\
      \V{y}(t) = \M{C}\,\V{x}(t)
    \end{cases}
\end{equation}
where $\M{A}$ is given by 
\begin{equation}
    \M{A} \triangleq -\i\, (\eye\otimes \Ham - \Ham^T\otimes\eye)
\end{equation}
and $\M{C}$ is given by
\begin{equation}
    \M{C} \triangleq
    \begin{bmatrix}
        \Vec{\Meas_1}^\dagger\\
        \Vec{\Meas_2}^\dagger\\
        \vdots\\
        \Vec{\Meas_{\numMeas}}^\dagger
    \end{bmatrix}\in \C^{\numMeas\times d^2}.
\end{equation}
The system is linearly observable if and only if the observability matrix
\[
    \M{O}(\M{A},\,\M{C}) \triangleq
    \begin{bmatrix}
        \M{C}\\
        \M{C}\M{A}\\
        \M{C}\M{A}^2\\
        \vdots\\
        \M{C}\M{A}^{d^2-1}
    \end{bmatrix} \in\C^{\numMeas d^2\times d^2}
\]
is of rank $d^2$ (i.e., $\M{O}(\M{A},\M{C})$ has a trivial nullspace)~\cite[pg. 221]{antsaklis}. An equivalent definition of observability is the ability to choose $\M{K}\in\C^{d^2 \times \numMeas}$ such that the matrix $(\M{A}-\M{K}\M{C})$ has eigenvalues in any desired location~\cite{luenberger}. The matrix $\M{K}$ is called the observer gain. A linear quantum state observer dynamically updates its estimate $\tilde{\V{x}}\in\C^{d^2}$ of the true state $\V{x}$ as follows:
\begin{equation}\label{vectorizedObs}
    \begin{cases}
      \dot{\tilde{\V{x}}}(t) = \M{A}\, \tilde{\V{x}}(t) + \M{K}(\V{y}(t) - \tilde{\V{y}}(t))\\
      \tilde{\V{y}}(t) = \M{C}\, \tilde{\V{x}}(t)
    \end{cases}.
\end{equation}
Define the observer error to be $\V{e}(t) \triangleq \tilde{\V{x}}(t) - \V{x}(t)$. Differentiating this expression using~\eqref{linearSys} and~\eqref{vectorizedObs} reveals the error dynamics
\begin{equation}
    \dot{\V{e}}(t) = (\M{A}-\M{K}\M{C})\, \V{e}(t).
\end{equation}
If one chooses $\M{K}$ such that all of the eigenvalues of $(\M{A}-\M{K}\M{C})$ are in the left-hand side of the complex plane with real components upper bounded by $-\lambda<0$, then the observer error tends towards zero exponentially, i.e. there exists an $M\in\R$ such that 
\begin{equation}
    \|\V{e}(t)\|_2 \leq M\, e^{-\lambda t}, \;\;\; \forall t\geq0.
\end{equation}
Letting $\hdens(t) = \Vec{\tilde{\V{x}}(t)}^{-1}$ and recalling that the vectorization operator is an isometry between $\C^{d^2}$ and $\C^{d\times d}$, one can conclude
\begin{equation}
    \|\hdens(t) - \dens(t)\|  = \|\V{e}(t)\|_2 \leq M\, e^{-\lambda t}
\end{equation}
for all $t\geq0$, which achieves the goal of~\eqref{convergentError}.

Linear observability cannot be achieved with an arbitrarily small number of observables. In fact, at least $d$ observables are required. The following result extends that of~\cite{clouatre2022linear} to the case where the measurement statistics do not come from a~\ac{POVM}.

\begin{thm}\label{thm:observability}
    If the closed quantum system $(\M{A},\, \M{C})$ is linearly observable then $\numMeas\geq d$.
\end{thm}
\begin{proof}
    The nullspace of $\M{A}$ has dimension at least $d$. Let $\V{v}_1,\V{v}_2,\dots,\V{v}_d$ be linearly independent vectors from this space, and define 
    \[
        \V{v}_j' \triangleq \M{C}\, \V{v}_j \in \C^{\numMeas}, \;\;\; j=1,2,\dots, d.
    \]
    If $\numMeas< d$, then the vectors $\{\V{v}_j'\}_{j=1}^d$ are linearly dependent and there exists a set of non-zero coeffecients $\{\beta_1,\beta_2,\dots,\beta_d\}$ such that
    \[
        \sum_{j=1}^d \beta_j\, \V{v}_j' = 0.
    \]
    Therefore, $\V{v}\triangleq \sum_{j=1}^d \beta_j\, \V{v}_j\neq \V{0}$ is such that $\M{C}\V{v}=\V{0}$ and $\M{C}\M{A}^k\V{v}=\V{0}$ for $k=1,2,\dots, d^2-1$. That is the nullspace of $\M{O}(\M{A},\, \M{C})$ contains a non-zero vector. The rank of the observability matrix is then strictly less than $d^2$ and the system is not linearly observable.
\end{proof}

One should note that $\hdens(t)$ is not necessarily a valid quantum density operator despite converging exponentially fast to an element of the set $\S$. One can rectify this by letting $\hhdens(t) \triangleq \proj_{\S}(\hdens(t))$ be the projection of $\hdens(t)$ onto the set $\S$ of valid density operators. Recently, a closed-form solution for this projection has been presented in the literature~\cite{clouatre2023closedform}. Because $\proj_{\S}$ is the Hilbert metric projection onto a closed convex set, it is non-expansive, i.e.  the projection cannot increase the estimation error. Therefore, $\hhdens(t)$ converges exponentially to the true state $\dens(t)$ of the system while also being a valid density:
\begin{equation}
    \|\hhdens(t) - \dens(t)\| \leq \|\hdens(t) - \dens(t)\| \leq M\, e^{-\lambda t}, \;\;\; \forall t\geq0.
\end{equation}

\begin{rem}
    The projection $\proj_{\S}$ is nonlinear since it does not obey the principle of superposition. For instance, for two matrices $\M{A}$ and $\M{B}$ it is never the case that $\proj_{\S}(\M{A}+\M{B})= \proj_{\S}(\M{A}) + \proj_{\S}(\M{B})$ since the left side of the equation has trace one and the right side has trace two. Hence, the estimate $\hhdens(t)$ is actually an exponentially convergent \textit{nonlinear} observer despite being based on linear observer theory.
\end{rem}

%

\section{Estimation of von Neumann \& relative entropy\label{sec:Entropy}}
Claude Shannon introduced the idea of the \emph{entropy}~\cite{Shan:48} of a random variable: given a random variable $\V{x}$ with (finite) probability distribution
\begin{equation}
    (P_1, P_2, \ldots, P_n) = (p(x_1),p(x_2),\ldots, p(x_n))
\end{equation}
the \emph{Shannon entropy} is
\begin{equation}
    Q(\V{x})=Q(P_1,P_2, \ldots, P_n) \triangleq -\sum_{i=1}^{n} p_i \log p_i.
\end{equation}
This entropy can be seen as both the average information gain when we learn the value of $x$, and, the average amount of uncertainty before we learn the value of $x$. The definition $0\log{0}\triangleq0$ is used for the value of the entropy at the origin. Later in \cite{von_neumannn}, von Neumann introduced the \emph{von Neumann quantum entropy}
\begin{align}
S(\dens)\triangleq - \tr{\dens \log \dens}
\end{align}
of the density $\dens\in\S$. In this case, $S$ measures the ``uncertainty'' of a quantum density operator. A pure state has entropy zero; however, the entropy of a mixed state is non-zero since it represents an ensemble of systems in various states.

Because $\dens$ is Hermitian, it can be diagonalized by a unitary $\M{Q}$ such that $\dens = \M{Q} \M{\Lambda} \M{Q}^\dagger$. Using this fact and standard properties of the trace and matrix logarithm, one can show
\begin{equation}
    S(\dens) = S(\M{\Lambda}) = -\sum_{k=1}^{d} \lambda_k \log \lambda_k
\end{equation}
which is also the Shannon entropy
\begin{equation}
    Q(\V{\lambda})\equiv -\sum_{k=1}^{d} \lambda_k \log \lambda_k
\end{equation}
given $\V{\lambda} = \diag\{\M{\Lambda}\}$. It is known, c.f.  \cite{nielsen}, that $S(\dens)$ is nonnegative and bounded:
\begin{align}
    0 \leq S(\dens) \leq \log d.
\end{align}
\begin{defn}
    A bounded function $f:\R_+\rightarrow\H$ is said to be \textit{essentially exponentially convergent} if there exists a finite time $T\in(0,\infty)$ and positive constants $M$ and $\sigma$ such that
    \[
        \|f(t)\| \leq Me^{-\sigma t}, \;\;\; \forall t\geq T.
    \]
\end{defn}

Using the notation of the previous section, we will prove the first main result of this note. This result shows that the nonlinear observer developed from control theory in Section~\ref{sec:LQSO} will also produce an essentially exponentially convergent estimate of the von Neumann entropy of $\dens$. The proof of Theorem~\ref{thm:thmVI1} will follow the presentation of two lemmas necessary for the proof.
\begin{thm}\label{thm:thmVI1}
    Let $\hhdens(t)$ be the density estimate produced by the quantum state observer described in the prior section. If the observer gain $\M{K}$ has been designed such that the observer error is bounded by the exponential rate $Me^{-\sigma t}$, then the von Neumann entropy of the estimated state $\hhdens(t)$ is essentially exponentially convergent to that of the true state $\dens$.
\end{thm}

\begin{lem}[Fannes' Inequality,~\cite{NieChu:00,Fannes1973}]\label{lem:Fannes}
    If $\dens,\altdens\in\S$ satisfy $\epsilon \triangleq \|\dens-\altdens\|_1 \leq \tfrac{1}{e}$, where $\|\cdot\|_1$ is the trace norm, then
    \begin{equation}\label{FannesIneq}
        |S(\dens) - S(\altdens)| \leq \epsilon\log d - \epsilon \log \epsilon.
    \end{equation}
\end{lem}

\begin{lem}\label{lem:myLemma}
    Let $\epsilon>0$ and $a\triangleq 1- \tfrac{1}{e}$. For all $t\geq0$, the following inequality holds
    \begin{equation}\label{myIneq}
        t e^{-\epsilon t} \leq \tfrac{1}{\epsilon}\, e^{- a \epsilon t}.
    \end{equation}
    This inequality is tight with equality holding when $t=\tfrac{1}{\epsilon(1-a)}$.
\end{lem}
\begin{proof}
    Note that the inequality trivially holds when $t=0$. Hence, let $t>0$. The inequality in question holds if and only if the inequality
    \[
        \mu e^{-\mu} \leq e^{-a\mu}
    \]
    holds with $\mu\triangleq \epsilon t>0$. However, this inequality is equivalent to
    \begin{eqnarray}\label{intIneq}
        1 & \leq & \frac{1}{\mu}\; e^{(1-a)\mu} \nonumber\\
          & = & \frac{(1-a)}{(1-a)\mu}\; e^{(1-a)\mu} \nonumber\\
          & = & (1-a) \frac{e^x}{x}
    \end{eqnarray}
    where $x\triangleq (1-a)\mu$. The function $f(x) = \frac{e^x}{x}$ on the domain $(0,\infty)$ is convex. Therefore, when $f'(x)=0$ then $f(x)$ is minimized. This occurs for $x=1$. Thus~\eqref{intIneq} holds when $(1-a)\geq\tfrac{1}{e}$, which is assumed in the statement of the lemma. Therefore inequality~\eqref{myIneq} is true. Equality holds when $t=\tfrac{1}{\epsilon(1-a)}$.
\end{proof}

\begin{proof}[Proof of Theorem~\ref{thm:thmVI1}.]
Recall that for any matrix $\M{A}\in\C^{d\times d}$ the trace norm $\|\M{A}\|_1$ and Hilbert-Schmidt norm $\|\M{A}\|$ satisfy the inequality $\|\M{A}\|_1 \leq \sqrt{d} \|\M{A}\|$.  By assumption of the theorem, the error dynamics of the observer are exponentially convergent.  Therefore $\epsilon(t) \triangleq \|\dens(t) - \hhdens(t)\|_1 \leq \sqrt{d}M e^{-\sigma t}$. Accordingly, for $T=\tfrac{\ln(eM\sqrt{d})}{\sigma}$ the inequality $\sqrt{d}M e^{-\sigma t} \leq 1/e$ holds true for all $t\geq T$. Past that point in time, Fannes' inequality holds:
\[
    |S(\dens) - S(\hhdens)| \leq \epsilon(t)\log(d) - \epsilon(t) \log(\epsilon(t)).
\]
Note that the function $g(x)\triangleq -x\log x$ is increasing on $(0,\tfrac{1}{e})$. Thus the convergence hypothesis pairs with Fannes' inequality to give
\begin{eqnarray*}
    |S(\dens) - S(\hhdens)| & \leq & \log(d) \sqrt{d}M e^{-\sigma t} \\
    && \> - \sqrt{d}Me^{-\sigma t} \log(\sqrt{d}Me^{-\sigma t})\\
    & = & \log(d) \sqrt{d}M e^{-\sigma t}\\
    && \> - \sqrt{d}M\left( \log(\sqrt{d}M) - \sigma t \right)e^{-\sigma t}.
\end{eqnarray*}
Using Lemma~\ref{lem:myLemma},
\begin{eqnarray*}
    |S(\dens) - S(\hhdens)| & \leq &  \log(d) \sqrt{d}M e^{-\sigma t} \\
    && \> - \sqrt{d}M\log(\sqrt{d}M) e^{-\sigma t} + \sqrt{d}M e^{-a \sigma t}.
\end{eqnarray*}
This proves the theorem.
\end{proof}

Von Neumann also defined the \emph{relative entropy} of two quantum densities:
\begin{align}
0 \leq S(\dens \| \altdens) \triangleq \text{tr}(\dens \log \dens - \dens \log \altdens).
\end{align}
This relative entropy is also the \emph{distinguishability} between $\dens$ and $\altdens$: when $S(\dens || \altdens)=0$, the two densities are identical~\cite[ch. 11.3]{NieChu:00}. Moreover, the relative entropy is unbounded. That is $S(\dens\|\altdens) \rightarrow +\infty$ is allowed. This happens, for instance, when the kernel of $\altdens$ is not contained within that of $\dens$. Thus, in general, we should not hope for convergence of the relative entropy. This is elucidated in the following.

\begin{rem}
    Consider the stationary density $\dens \triangleq \ketbra{0}{0}$ alongside the time-dependent density
    \[
        \altdens(t) \triangleq
        \begin{bmatrix}
            1-e^{-t} & 0\\
            0 & e^{-t}
        \end{bmatrix}.
    \]
    Note that $\altdens(t) \rightarrow \dens$ exponentially as $t\rightarrow\infty$. However, at any finite time $t\in(0,\infty)$ the relative entropy $S(\altdens\|\dens) \triangleq +\infty$. This example shows that convergence of the relative entropy is a strictly stronger condition than convergence in the Hilbert-Schmidt norm.
\end{rem}

\begin{figure*}[t!]
    \centering
    \subfloat[]{\hspace{-6mm}\centering\includegraphics[width=0.475\textwidth]{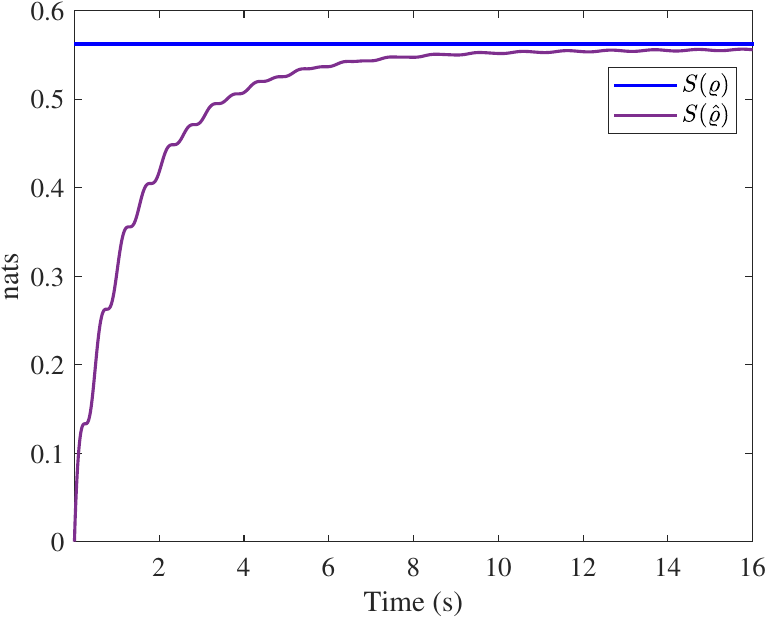}}
    \hfill
    \subfloat[]{\hspace{-6mm}\centering\includegraphics[width=0.475\textwidth]{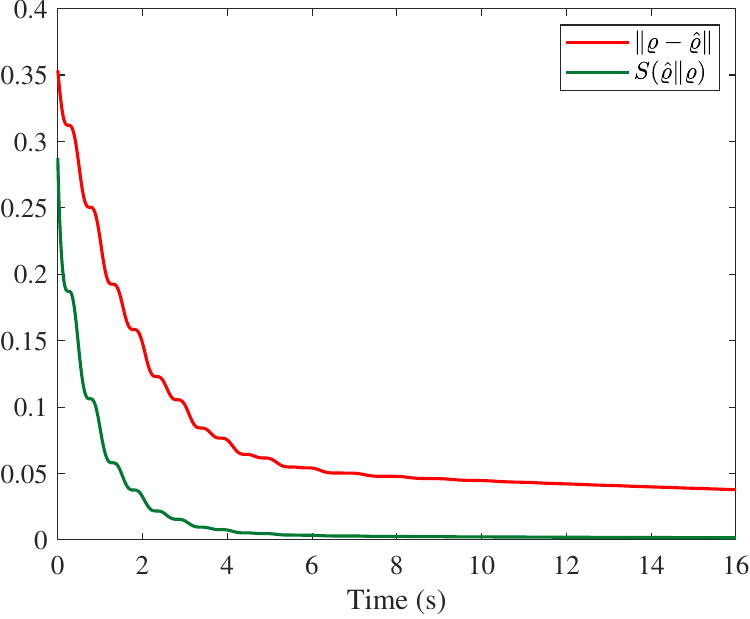}}
    \caption{(a) shows the entropy $S(\hhdens)$ estimated using a quantum state observer converging to that of the system $S(\dens)$ as the observer's error converges to zero. (b) shows the relative entropy $S(\hhdens\|\dens)$ converging to zero as $\|\dens-\hhdens\|$ converges to zero.}
    \label{fig:1}
\end{figure*}

Despite this example, under appropriate assumptions, we can prove that the relative entropy is convergent. The second main result of this section is presented below.

\begin{thm} \label{thm:thmv14}
    Suppose that $\dens(0)$ is positive definite. Under the same hypothesis of Theorem \ref{thm:thmVI1}, the relative entropy $|S(\hhdens\| \dens)|$ is essentially exponentially convergent.
\end{thm}

\begin{proof}
    Given $|S(\hhdens||\dens)| = |\tr{ \hhdens \log \hhdens} - \tr{\hhdens \log \dens}|$,
    \begin{align*}
        |S(\hhdens\| \dens)| &= |\tr{ \hhdens \log \hhdens} + S(\dens) - S(\dens) - \tr{ \hhdens \log \dens}| \\
        &\leq |S(\dens) - S(\hhdens)| + |\underbrace{\tr{(\dens- \hhdens) \log \dens}}_{(\dens-\hhdens,\log \dens)_{\text{tr}}}|.
    \end{align*}
    Using the Cauchy–Schwarz inequality,
    \begin{align*}
        |S(\hhdens||\dens)| &\leq |S(\dens) - S(\hhdens)|+ \|\dens - \hhdens\| \cdot \|\log \dens\|.
    \end{align*}
    Note that $\dens(t) = \M{U}(t) \dens(0) \M{U}(t)^\dagger$ where $\M{U}(t) \triangleq e^{-\i \Ham t}$ is unitary. Applying the definition of the matrix logarithm, 
    \[
        \log \dens(t) = \M{U}(t) \, \log( \dens(0)) \, \M{U}(t)^\dagger.
    \]
    Combining this equation with the unitary invariance of the Hilbert-Schmidt norm on $\H$, we see that $\|\log \dens(t)\| = \|\log \dens(0)\|$ for all $t\geq0$. Since $\dens(0)$ is assumed to be positive definite  $D\triangleq \|\log \dens(0)\| < \infty$. Therefore
    \begin{align*}
        S(\hhdens\| \dens) \leq |S(\dens) - S(\hhdens)| + D \|\hhdens-\dens\|.
    \end{align*}
    Theorem \ref{thm:thmVI1} proved that $|S(\dens) - S(\hhdens)|$ is essentially exponentially convergent.  Hence $S(\hhdens\| \dens)$ is the sum of an essentially exponentially convergent term and an exponentially convergent term.  Therefore $S(\hhdens\| \dens)$ is itself essentially exponentially convergent.
\end{proof}

Theorem \ref{thm:thmv14} shows that the quantum state observer developed in Section \ref{sec:LQSO} will also produce essentially exponential convergence of the relative entropy to zero as long as $\dens(0)$ is full-rank. One would imagine the rank condition is satisfied quite often in nature; however, it is an interesting open research problem to study if this assumption can be relaxed.

\subsection{Example: A laser-driven atom}
A numerical example will now be presented to illustrate the results developed in this note. Consider a laser-driven atom governed by the Liouville-von Neumann equation with Hamiltonian
\[
    \Ham = 
    \begin{bmatrix}
        E_0 & \omega\\
        \bar{\omega} & E_1
    \end{bmatrix}
\]
where $E_0,E_1\in\R$ are its energy eigenvalues and $\omega\in\C$ is the (constant) driving frequency. Assume $E_0=-0.5$, $E_1=0.5$, and $\omega=3$. The output $\V{y}(t)$ will consist of the expected value of the projective measurements onto the eigenstates of the undriven Hamiltonian: $\ketbra{0}{0}$ and $\ketbra{1}{1}$. After vectorizing the master equation, one will find that the Kalman observability matrix $\M{O}(\M{A},\M{C})$ has rank 4. Thus the system is linearly observable. The observer gain $\M{K} = \M{C}^\dagger$ will be used, which ensures the exponential stability of the observer error dynamics. The system and observer are initiated in the states
\[
    \dens_0 \triangleq
    \begin{bmatrix}
        0.25 & 0\\
        0 & 0.75
    \end{bmatrix}
    \;\;\; \text{and} \;\;\;
    \hhdens_0 \triangleq
    \begin{bmatrix}
        0 & 0\\
        0 & 1
    \end{bmatrix}.
\]
Figure~\ref{fig:1} plots the results of this numerical experiment. In Figure~\ref{fig:1}a, one can see that the von Neumann entropy of the estimated state $\hhdens$ converges to that of the true state $\dens$ as guaranteed by Theorem~\ref{thm:thmVI1}. In Figure~\ref{fig:1}b, one can see that the relative entropy $S(\hhdens\|\dens)$ converges exponentially to zero as the normed estimation error converges exponentially to zero. This is when the quantum states $\dens$ and $\hhdens$ become indistinguishable. This was guaranteed by Theorem~\ref{thm:thmv14} since $\dens_0$ is full-rank.

\vspace{5mm}
\section{Conclusion}
A valid quantum state observer is one whose state converges to that of the reference quantum system while also ensuring that its estimated state is a valid density operator. This note showed that when a valid quantum state observer is used to infer the state of a closed quantum system: (i) the entropy of the observer's state is always essentially exponentially convergent to that of the system's state, and (ii) the relative entropy between the observer and system's states is essentially exponentially convergent to zero as long as the system starts in a full-rank state.



\bibliography{
    BibFiles/ref.bib,
    BibFiles/controlTextbooks.bib,
    BibFiles/balasBiblio.bib,
    BibFiles/clouatreBiblio.bib,
    BibFiles/quantumTextbooks.bib,
    BibFiles/originalRefs.bib
}
\bibliographystyle{apsrev4-1}


\end{document}